\documentclass{article}
\pdfoutput=1
\usepackage{amssymb,amsmath}
\usepackage{wasysym}
\usepackage{graphicx}
\usepackage{epsfig}

\usepackage{latexsym}

\def \prend{\vrule depth-1pt height7pt width6pt}

\usepackage{url}

\newcommand{\seepage}[1]{\marginpar{\scriptsize (p.~\pageref{#1})}}

\if01
\newtheorem{theorem}{T\/heorem}[section]
\newtheorem{apptheo}{T\/heorem}[section]
\newtheorem{corollary}{Corollary}[section]
\newtheorem{definition}{Definition}[section]
\newtheorem{lemma}{Lemma}[section]
\newtheorem{applem}{Lemma}[section]
\newtheorem{example}{Example}[section]
\newtheorem{comment}{Comment}[section]
\newtheorem{fact}{Fact}[section]
\newtheorem{claim}{Claim}[section]
\newtheorem{proposition}{Proposition}[section]
\newtheorem{remark}{Remark}[section]

\newtheorem{open}{Open problem}[section]
\newtheorem{theorem}{T\/heorem}

\newtheorem{corollary}[theorem]{Corollary}
\newtheorem{definition}[theorem]{Definition}
\newtheorem{lemma}[theorem]{Lemma}

\newtheorem{example}[theorem]{Example}

\newtheorem{claim}[theorem]{Claim}
\newtheorem{proposition}[theorem]{Proposition}

\newtheorem{open}[theorem]{Open problem}

\newtheorem{open}{Open problem}[section]
\fi

\newcommand{\rfig}[1]{Fig.~\ref{#1}}

\newtheorem{theorem}{Theorem}[section] 
\newtheorem{lemma}[theorem]{Lemma}
\newtheorem{proposition}[theorem]{Proposition}
\newtheorem{corollary}[theorem]{Corollary}
\newtheorem{definition}[theorem]{Definition}

\newtheorem{open}[theorem]{Open problem}
\newtheorem{example}[theorem]{Example}

\begin{document}


\title{Outfix-Guided 
	Insertion\thanks{
This   manuscript   version   is   made   available   under   the   
CC-BY-NC-ND   4.0   license
{\tt http://creativecommons.org/licenses/by-nc-nd/4.0/}
Accepted in Theoretical Computer Science
{\tt http://dx.doi.org/10.1016/j.tcs.2017.03.040}
	An 
extended abstract of this paper appeared in the {\it Proceedings of the
	20th International Conference Developments
in Language Theory,} DLT 2016, Lect. Notes Comput. Sci. 9840,
Springer-Verlag, 2016, pp. 102--113.}}

\date{February 20, 2017}


\author{Da-Jung Cho\thanks{Department of Computer Science, Yonsei University,
50, Yonsei-Ro, Seodaemum-Gu, Seoul 120--749, Republic of Korea
{\tt \{dajungcho, emmous\}@yonsei.ac.kr}} 
\and Yo-Sub Han$^{\dag}$ \and Timothy 
Ng\thanks{School of Computing, Queen's University,
Kingston, Ontario K7L 3N6, Canada
{\tt \{ng, ksalomaa\}@cs.queensu.ca}} \and Ka Salomaa$^{\ddag}$}


%
%
%

\maketitle

\begin{abstract}
Motivated by work on bio-operations on DNA strings,
we consider an outfix-guided insertion operation
that can be viewed as a generalization of the overlap assembly
operation on strings studied previously.
As the main result we construct a finite language $L$
such that the outfix-guided insertion closure of $L$ is
non-regular. We consider also the closure properties of regular and
(deterministic) context-free languages under the outfix-guided
insertion operation and  decision problems
related to outfix-guided insertion. Deciding whether a language
recognized by a deterministic finite automaton is closed under
outfix-guided insertion can be done in polynomial time.
The complexity of the corresponding question for nondeterministic
finite automata remains open.

\end{abstract}

{\bf Keywords:}
	language operations, closure properties, regular languages


\section{Introduction}

Gene insertion and deletion are basic operations occurring
in DNA recombination in molecular biology.
Recombination  creates a new  DNA strand by cutting, 
substituting, inserting, deleting or combining other strands.
Possible  errors in this process
impair the function of genes.
Errors in DNA recombination cause mutation that plays a 
part in normal and abnormal biological processes such as cancer, 
the immune system, protein synthesis and evolution~\cite{Bertram00}.
Since mutational damage may or may not 
be easily identifiable, researchers deliberately 
generate mutations so that the structure and biological activity of 
genes can be examined in detail.
\emph{Site-directed mutagenesis} is 
one of the most important techniques in laboratory
 for generating mutations on specific sites of DNA using PCR 
(polymerase chain reaction) based methods~\cite{FlavellSBW75,HemsleyATDCG89}.
For a site-directed insertion mutagenesis by PCR, the mutagenic primers are typically designed to include the desired change, which could be base addition~\cite{LeeSRKR10,LiuN08}.
This enzymatic reaction occurs in 
the test tube with a DNA strand and predesigned primers 
in which the DNA strand includes a target region, 
and a predesigned primer includes a complementary region of the target region.  
The complementary region of primers leads it to hybridize the target DNA region 
and generate a desired insertion on a specific site as a mutation.
\rfig{fig:mutagenesis} illustrates the
procedure of site-directed insertion mutagenesis by PCR.
\begin{figure}
\includegraphics[width=12cm]{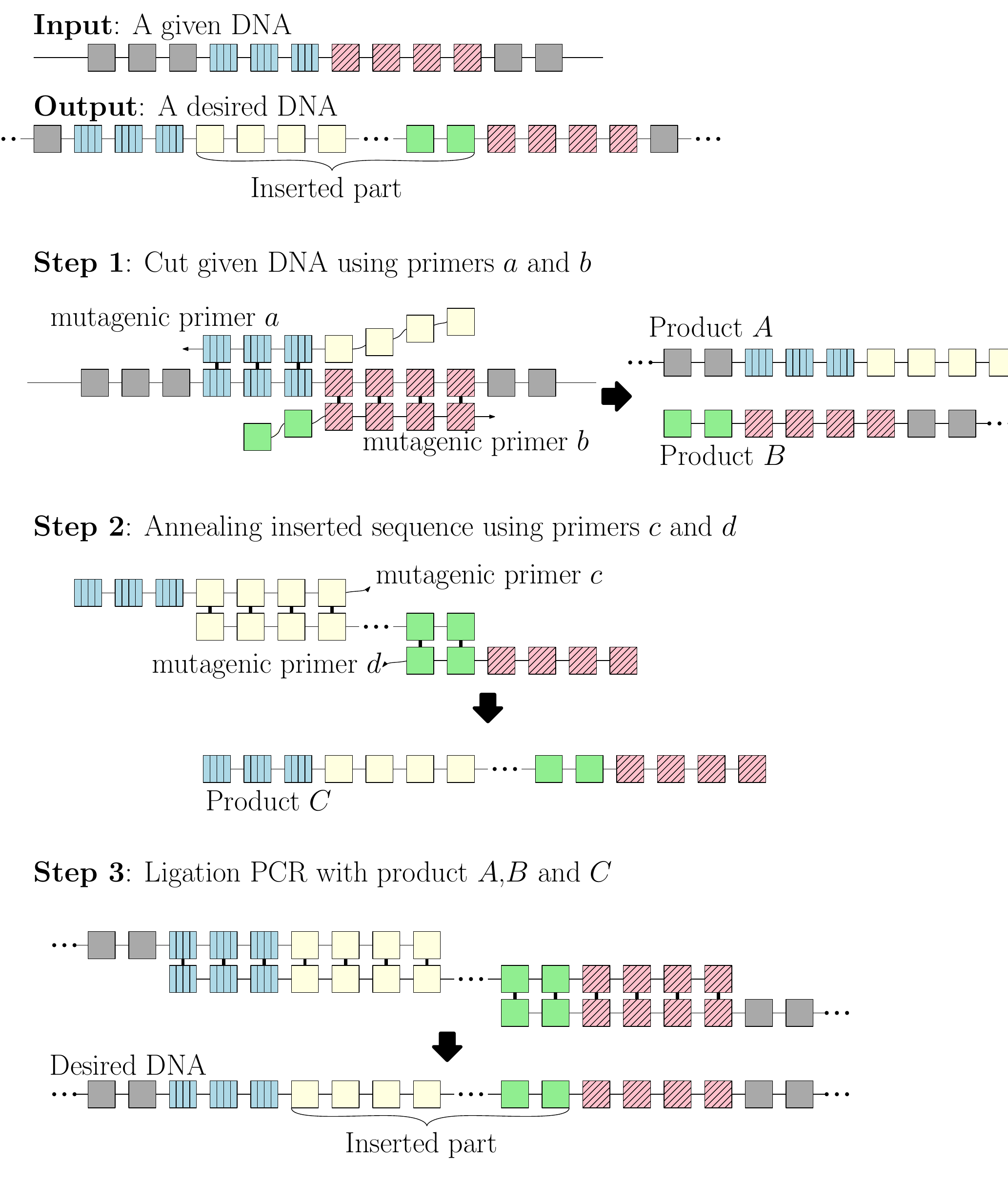}
\caption{An example of site-directed insertion mutagenesis by PCR. Given a DNA sequence and four predesigned primers $a,b,c$ and $d$, two primers~$a$ and $b$ lead the DNA sequence to break and extend into two products~$A$ and $B$ under enzymatic reaction (Step 1). Two primers~$c$ and $d$ complementarily bind to desired insertion region according to the overlapping region and extend into product~$C$ (Step 2). Then, the products~$A,B$ and $C$ join together to create recombinant DNA that include the desired insertion (Step 3).}
\label{fig:mutagenesis}
\end{figure}

In formal language theory, the insertion of a string means
adding a substring to a given string and deletion of a string
means removing a substring. The insertions occurring in
DNA strands are in some sense context-sensitive 
and 
Kari and Thierrin~\cite{KariT96} modeled
such bio-operations using {\em contextual insertions
and deletions\/}~\cite{Galiukschov1981,Paun84}. 
A finite set of insertion-deletion rules, together
with a finite set of axioms, can be viewed as
a language generating device. Contextual insertion-deletion
systems in the study of molecular computing have been
used e.g. by Daley et al.~\cite{DaleyKGS99}, 
Enaganti et al.~\cite{Enaganti2}, 
Krassovitskiy 
et al.~\cite{KrassobitskiyRV11} and Takahara and
Yokomori~\cite{TakaharaY03}. Further theoretical studies on the
computational power of insertion-deletion systems
were done e.g. by Margenstern et al.~\cite{MargensternPRV05} and
P\v{a}un et al.~\cite{PaunPY08}.
Enaganti et al.~\cite{Enaganti2} have studied related operations
to model the action of DNA polymeraze enzymes.

We formalize site-directed insertion mutagenesis by PCR 
and define a new operation \emph{outfix-guided insertion}  that 
{\em partially\/} inserts 
a string $y$ into a string $x$ when two non-empty substrings of 
$x$ match with an outfix of $y$, see~\rfig{fig:kaksi}~(b).  
We will consider also variants where only a prefix or a suffix of $y$
must match with a non-empty substring of $x$ at the position where
the insertion occurs.
The outfix-guided insertion is an overlapping variant of the ordinary
insertion operation, analogously as the overlap assembly
\cite{Csuhaj2007,EnagantiIKK15,Enaganti2016}, 
 cf.~\rfig{fig:kaksi}~(a), is a 
 variant of the ordinary string concatenation operation. 
An operation equivalent to overlap assembly
has been considered under the name
 chop of languages by Holzer et al.~\cite{Holzer1}. 
 Holzer and Jacobi~\cite{Holzer2}
have given tight state complexity bounds for 
a variant of the operation where the overlapping string always
has length one.
Furthermore, C\v{a}r\v{a}u\c{s}u and P\v{a}un~\cite{CarPaun} have considered
 another related operation called short concatenation.

\begin{figure}
\centering
\includegraphics[width=11cm]{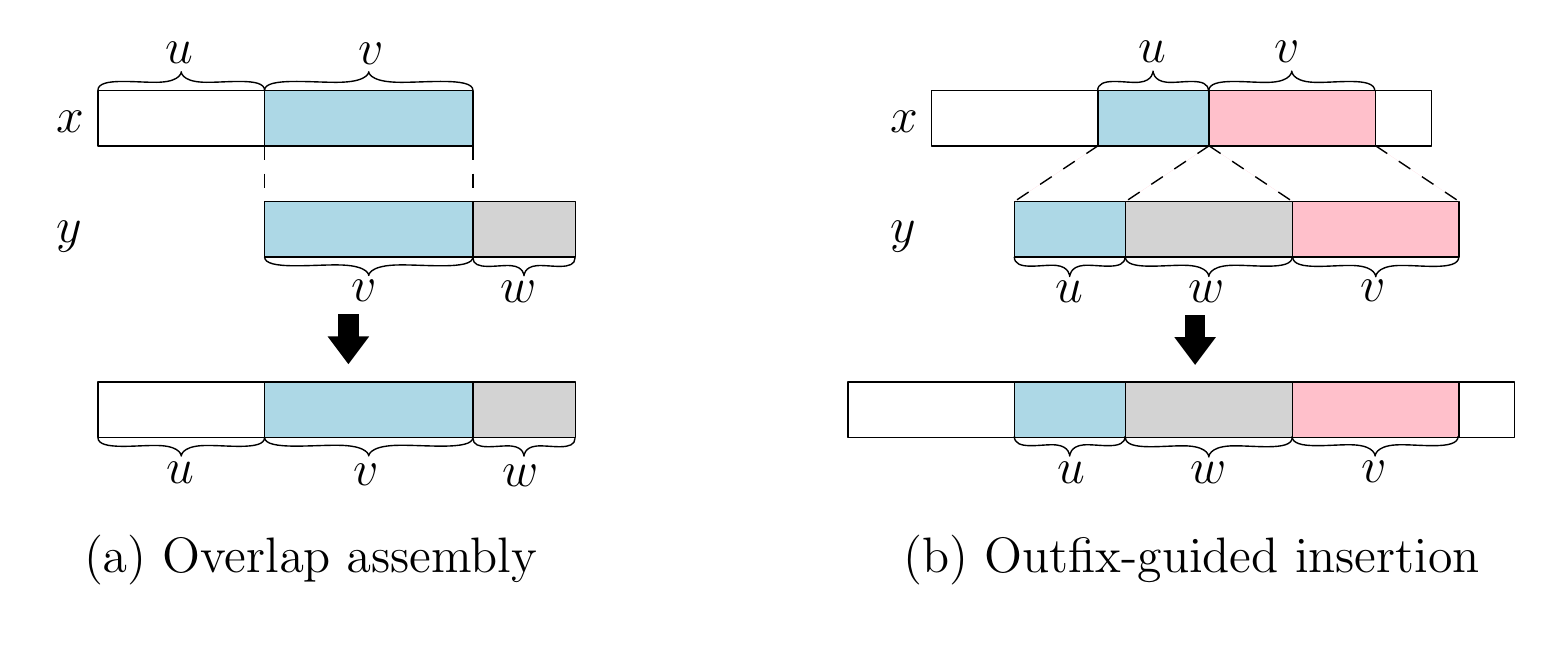}
\caption{
(a) If suffix $v$ of $x$ overlaps with prefix $v$ of $y$, then the
overlap assembly operation appends suffix $w$ of $y$ to $x$.
(b) If the outfix of $y$ consisting of~$u$ and $v$ matches 
the substring~$uv$ of $x$, then the outfix-guided insertion operation 
inserts $w$  between $u$ and $v$ in the string $x$.}
\label{fig:kaksi}
\end{figure}

This paper investigates the language
theoretic closure properties of outfix-guided insertion and 
iterated outfix-guided insertion. Note that since outfix-guided
insertion, similarly as overlap assembly, is not associative,
there are more than one way to define the iteration of the operation.
We consider a general outfix-guided insertion closure of a language
which is defined analogously as the iterated overlap 
assembly by Enaganti et al.~\cite{EnagantiIKK15}.
Iterated (overlap)  assembly is defined by
Csuhaj-Varju et al.~\cite{Csuhaj2007} in a different way, which we call
right one-sided iteration of an operation.

It is fairly easy to see that regular languages are closed under
outfix-guided insertion. Closure of regular languages under
iterated outfix-guided insertion  turns out to be less obvious.
It is well known that regular languages are not closed under
the iteration of the ordinary (non-overlapping) insertion
operation~\cite{Kari91} and it is also fairly easy to establish
that iterated prefix-guided (or suffix-guided) insertion
does not preserve regularity. However,  the known 
counter-examples, nor their variants, do not work for iterated
outfix-guided insertion. Here using a more involved construction
we show that there exists even a finite language $L$ such that
the outfix-guided insertion closure of $L$ is non-regular.
On the other hand, we show that the outfix-guided insertion closure
of a unary regular language is always regular.

It is well known that context-free languages are closed under
ordinary (non-iterated) insertion. We show
that context-free languages are not closed under outfix-guided
insertion, nor under prefix-guided
or suffix-guided insertion. 
The outfix-guided insertion of a regular language into
a context-free language (or vice versa) is always context-free.
Also we establish that a similar closure property does not hold
for the deterministic context-free  and the regular languages.
Finally in section~6 we consider decision problems on
whether a language is closed under outfix-guided insertion (or 
og-closed).
We give a polynomial time algorithm to decide whether a language
recognized by a deterministic finite automaton (DFA) is
og-closed. We show that for a given context-free language $L$ the
question of deciding whether or not $L$ is og-closed is undecidable.

\section{Preliminaries}

We assume the reader to be familiar with the basics of formal
languages, in particular, with the classes
of regular languages and (deterministic)
context-free languages~\cite{Shallit09,Yu97}. 
Here we briefly recall some definitions and in the next
section formally define the  the main notion of outfix-guided insertion
and the corresponding iterated operations.

The symbol $\Sigma$ stands always for a finite alphabet, 
$\Sigma^*$ (respectively, $\Sigma^+$) is the set of strings
(respectively, non-empty strings) over $\Sigma$,
$|w|$ is the length of a string $w \in \Sigma^*$, 
$w^R$ is the reversal of $w$ and
$\varepsilon$ is the empty string.
For $i \in \mathbb{N}$,
$\Sigma^{\geq i}$ is the set of strings of length at least $i$.

If $w = xy$, $x, y \in \Sigma^*$, we say that $x$ is a 
{\em prefix\/} of $w$ and $y$ is a {\em suffix\/} of $w$. If
$w = xyz$, $x, y, z \in \Sigma^*$, we say that
$(x,z)$ is an {\em outfix\/} of $w$. If additionally
$x \neq \varepsilon$ and $z \neq \varepsilon$, 
$(x,z)$ is a {\em non-trivial outfix\/} of
$w$. 
Sometimes (in particular, when talking about the outfix-guided
insertion operation) we refer to
an outfix $(x, z)$ simply as a string $xz$ (when it is known from the
context what are the components $x$ and $z$).

\begin{example}
Let $\Sigma = \{ a, b, c \}$ and $w = abca$. The non-trivial
outfixes of $w$ are $(a, a)$, $(ab,a)$,  $(a,ca)$,
$(a, bca)$, $(ab, ca)$, and $(abc, a)$. 
Note
that all prefixes and suffixes of a string $u$ are outfixes of
$u$ but prefixes and suffixes are not, in general, non-trivial
outfixes. 
A string $u$ represents one or more
non-trivial outfixes of $u$ if and only
if $|u| \geq 2$.
\end{example}

To conclude this section we fix some basic notation
on finite automata.

A {\em nondeterministic finite automaton\/} (NFA) is a tuple
$A = (\Sigma, Q, \delta, q_0, F)$ where $\Sigma$ is
the input alphabet, $Q$ is the finite set of states,
$\delta \colon Q \times \Sigma \rightarrow 2^Q$ is
the transition function, $q_0 \in Q$ is the
initial state and $F \subseteq Q$ is the set of final states.
In the usual way $\delta$ is extended as a function
$Q \times \Sigma^* \rightarrow 2^Q$ and the {\em language accepted
by\/} $A$ is $L(A) = \{ w \in \Sigma^* \mid \delta(q_0, w) \cap F
\neq \emptyset \}$. The automaton $A$ is a {\em deterministic finite
automaton\/} (DFA) if $|\delta(q, a)| \leq 1$ for all $q \in Q$
and $a \in \Sigma$.

It is well known that the deterministic and nondeterministic
finite automata recognize the class of {\em regular languages.}
A (nondeterministic) {\em pushdown automaton\/} (PDA) is an
extension of a finite automaton that reads the input
left-to-right and in addition to the
finite  state memory has  access to a pushdown 
store~\cite{Shallit09}. The nondeterministic PDAs define
the class of {\em context-free languages\/} (CFL). 
Deterministic PDAs define
the class of  {\em deterministic context-free languages\/} (DCFL) and
this is a proper subclass of CFL \cite{Shallit09}.

\section{Definition of (Iterated) Outfix-Guided Insertion}


We begin by recalling some notions associated with the
non-overlapping insertion operation.\footnote{We use the
term ``non-overlapping'' to make the distinction clear to
outfix-guided insertion which will be the main topic of this paper.}
More details
on variants of the insertion operation and
iterated insertion can be found in~\cite{Kari91}.

The  non-overlapping
insertion of a string $y$ into a string $x$ is defined
as the set of strings $x \stackrel{\rm nol}{\leftarrow}
 y = \{ x_1 y x_2 \mid x = x_1 x_2 \}$.
The insertion operation is extended in the natural way to
languages by setting
$
L_1 \stackrel{\rm nol}{\leftarrow} L_2 = \bigcup_{x \in L_1, y \in L_2} x 
\stackrel{\rm nol}{\leftarrow} y.
$
Following Kari~\cite{Kari91} we define the {\em left-iterated
insertion \/} of $L_2$ into  $L_1$ inductively by
setting
$$
\mathbb{LI}^{(0)}(L_1, L_2) = L_1 \mbox{ and }
\mathbb{LI}^{(i+1)}(L_1, L_2) = \mathbb{LI}^{(i)}(L_1, L_2) 
\stackrel{\rm nol}{\leftarrow}
L_2, \; i \geq 0.
$$
The {\em left-iterated insertion closure of $L_2$ into $L_1$\/} is
$
\mathbb{LI}^*(L_1, L_2) = \bigcup_{i = 0}^{\infty} \mathbb{LI}^{(i)}(L_1, L_2).
$
It is well known that the iterated non-overlapping insertion operation does 
not preserve regularity \cite{Kari91,Haussler81}. 

\begin{example}
\label{etavallinen}
Let $\Sigma = \{ a, b \}$. The left-iterated insertion closure
of the string $ab$ into itself is non-regular because
 $\mathbb{LI}^*(ab, ab) \cap a^* b^* = 
\{ a^i b^i \mid i \geq 0 \}$.
\end{example}

Next we define the main notion of this paper which can be
viewed as a generalization of the overlap
assembly operation~\cite{Csuhaj2007,EnagantiIKK15}. 
An ``inside part'' of a string $y$ can be outfix-guided inserted
into a string $x$ if a
non-trivial outfix of $y$   overlaps with a substring
of $x$ in a position where the insertion occurs.
This differs from contextual insertion (as defined in~\cite{KariT96})
in the sense that $y$ must actually contain the outfix that is
matched with a substring of $x$ (and additionally \cite{KariT96}
specifies a set of contexts where an insertion can occur). 

\begin{definition}
\label{maindefin}
The {\em outfix-guided insertion\/} of a string $y$ into a string $x$
is defined as
$$
x \stackrel{\rm ogi}{\leftarrow} y = \{ x_1 u z v x_2 \mid x = x_1 u v x_2,
\; y = uzv, \; u \neq \varepsilon, v \neq \varepsilon \}.
$$
\end{definition}

Using the above notations, when
$w =  x_1 u z v x_2$ is the result of outfix-guided inserting
$y = uzv$ into $x = x_1 uv x_2$ we say
that the non-empty substrings
 $u$ and $v$ are the {\em matched parts.} Note that the
 matched parts form a non-trivial outfix
of the inserted string $y$.
When speaking of matched parts we refer to specific substring
occurrences in the string $x$ that are matched with a non-empty prefix
and suffix of $y$, respectively. When string $u$ occurs as
a substring $x$ after a prefix of length $i \in \mathbb{N}$, this
could be specified as a pair $(u, i+1)$ to indicate that the
occurrence begins at position $i+1$.

As  variants of outfix-guided insertion we define operations where
only a non-empty  prefix or a non-empty
  suffix of the inserted string needs to be
matched with a substring in the original string. Naturally
it would be possible to define further variants of outfix-guided
insertion, e.g., by allowing the matched outfix to be empty.

\begin{definition}
\label{prefsuf}
The {\em prefix-guided insertion\/} of a string $y$ into a string $x$
is defined as
$$
x \stackrel{\rm pgi}{\leftarrow} y = \{ x_1 y_1 y_2 x_2 \mid x = x_1 y_1 x_2,
\; y = y_1 y_2, \; y_1 \neq \varepsilon \}.
$$
The {\em suffix-guided insertion\/} of a string $y$ into a string $x$
is defined as
$$
x \stackrel{\rm sgi}{\leftarrow} y = \{ x_1 y_1 y_2 x_2 \mid x = x_1 y_2 x_2,
\; y = y_1 y_2, \; y_2 \neq \varepsilon \}.
$$
\end{definition}

The ordinary insertion, outfix-guided insertion and suffix-guided
insertion operations, respectively,
 are illustrated in Fig.~\ref{fig:OGI}.

\begin{figure}
\centering
\includegraphics[width=11cm]{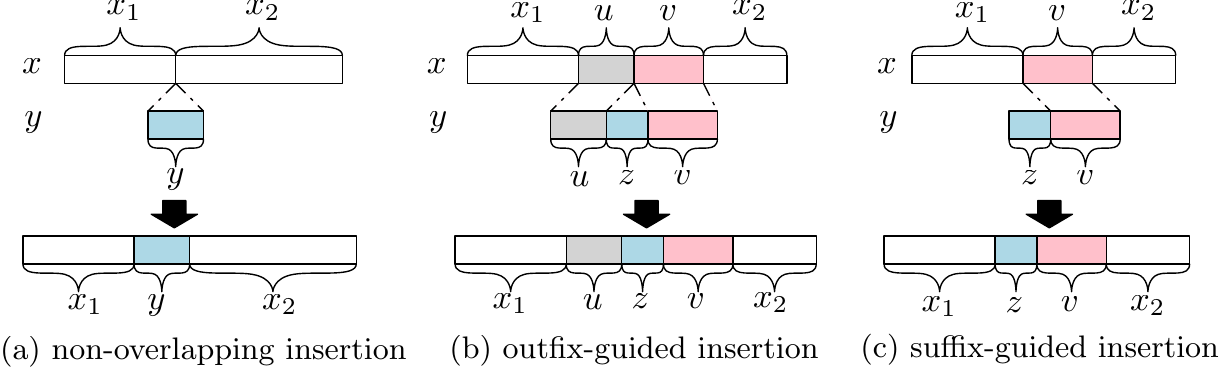}
\caption{
(a) Non-overlapping insertion of string $y$ into string $x$.
(b) If the outfix of $y$ consisting of~$u$ and $v$ matches 
the substring~$uv$ of $x$, then the outfix-guided insertion operation 
inserts $z$  between $u$ and $v$ in the string $x$.
(c) If the suffix of $y$ consisting of $v$ matches a substring of $x$,
then the suffix-guided operation inserts the prefix $z$ of $y$ before
an occurrence of $v$ in $x$.}
\label{fig:OGI}
\end{figure}

Since we are mainly dealing with outfix-guided
insertion, in the following for notational simplicity
we write just $\leftarrow$ in place of
$\stackrel{\rm ogi}{\leftarrow}$.
Outfix-guided insertion
is extended in the usual way for languages by setting
$
L_1 \leftarrow L_2 = \bigcup_{w_i \in L_i, i = 1, 2}
w_1 \leftarrow w_2.
$
The prefix-guided and suffix-guided insertion 
operations $\stackrel{\rm pgi}{\leftarrow}$  and
$\stackrel{\rm sgi}{\leftarrow}$ are extended
for languages in the same way.

It is known that the ordinary insertion operation is not associative
and, not surprisingly, neither are the outfix-, prefix- and 
suffix-guided variants.
\begin{example}
\label{eetta2}
Outfix-guided (respectively, prefix-guided, suffix-guided)
insertion operation is not associative. 

Let
$\Sigma = \{ a, b, c, d \}$. Now
$abcd \in (acd \leftarrow abc) \leftarrow abcd$ but
$abc \leftarrow abcd = \emptyset$.  

Similarly we note that 
$abc \in (ab \stackrel{\rm pgi}{\leftarrow} bc)
\stackrel{\rm pgi}{\leftarrow} abc$ but
$bc \stackrel{\rm pgi}{\leftarrow} abc = \emptyset$ because
no substring of $bc$  is a prefix of $abc$. By reversing all the strings we get
an example that shows that suffix-guided insertion is non-associative.
\end{example}

Since outfix-guided (prefix-guided,
suffix-guided, respectively) insertion is non-associative we define
the $(i+1)$st iterated operation, analogously as was done with iterated
overlap assembly \cite{EnagantiIKK15}, by inserting to a string
of the  $i$th iteration another string of the $i$th iteration.

\begin{definition}
\label{datta2}
For a language $L$  define inductively 
$$\mathbb{OGI}^{(0)}(L) = L,  \;\; \mbox{ and }
\mathbb{OGI}^{(i+1)}(L) = \mathbb{OGI}^{(i)}(L)
\leftarrow  \mathbb{OGI}^{(i)}(L), 
\;\;
i \geq 0.$$
The {\em  outfix-guided insertion closure\/} of $L$ is
$$
\mathbb{OGI}^*(L) = \bigcup_{i=0}^{\infty}  \mathbb{OGI}^{(i)}(L).
$$

The {\em prefix-guided insertion closure\/} of $L$, 
$\mathbb{PGI}^*(L)$,
(respectively, {\em suffix-guided insertion closure\/}
of $L$, $\mathbb{SGI}^*(L)$) is defined as above by replacing
$\leftarrow$ everywhere
with $\stackrel{\rm pgi}{\leftarrow}$ (respectively,
with $\stackrel{\rm sgi}{\leftarrow}$).
\end{definition}

Recall that the left-iterated non-overlapping insertion~\cite{Kari91} discussed
above uses two argument languages, and the same is true
for  the left- and right-iterated outfix-guided
insertion introduced below in Definition~\ref{datta5}.
One of the arguments can be viewed as the ``target'' of the insertions,
and the other as the ``source'' of the inserted strings.
The unrestricted insertion closures of
Definition~\ref{datta2} are defined for one argument language
because, roughly speaking, the $i+1$st stage uses the $i$th stage
both as the target and the source of the insertion.

For talking about specific iterated outfix-guided insertions,
we use the notation
$x \stackrel{[y]}{\Rightarrow} z$ to indicate
that string $z$ is in $x \leftarrow y$, $x, y, z \in \Sigma^+$.
A sequence of 
steps  
$$x \stackrel{[y_1]}{\Rightarrow} z_1
\stackrel{[y_2]}{\Rightarrow} z_2 \stackrel{[y_3]}{\Rightarrow}
\cdots \stackrel{[y_m]}{\Rightarrow} z_m, \;\; m \geq 1,$$ 
is called a {\em derivation\/} of $z_m$ from $x$.

When we want to specify the matched substrings, they
are  indicated by underlining. If $x = x_1 u v x_2$
derives $z$ by inserting $uyv$ (where $u$ and $v$ are the matched
prefix and suffix, respectively,) this is denoted
$$
x_1 \underline{u} \underline{v} x_2
\stackrel{[\underline{u} y \underline{v}]}{\Rightarrow} z.
$$
Also, sometimes underlining is done only in the inserted
string if this makes it clear what must be the matched
substrings in the original string.

By a {\em trivial derivation step\/} we mean a derivation
$x \stackrel{[x]}{\Rightarrow} x$ where $x$ is obtained from
itself by selecting the outfix to consist of the entire
string $x$. Every string of
length at least two can be obtained from itself using a trivial
derivation step. This means, in particular, that
for any language $L$,
$
L - (\Sigma \cup \{ \varepsilon \}) \subseteq \mathbb{OGI}^{(1)}(L).
$
The sets $\mathbb{OGI}^{(i)}(L)$, $i \geq 1$, cannot contain strings
of length less than two and, consequently 
$
\mathbb{OGI}^{(i)}(L) \subseteq \mathbb{OGI}^{(i+1)}(L)$, 
 for all  $i \geq 1$.

Definition~\ref{datta2} iterates the outfix-guided insertion by
inserting a string from the $i$th iteration of the operation into
another string in the $i$th iteration. 
 Since the operation is non-associative
we can define iterated insertion in more than one way. 
The  right one-sided iterated insertion of $L_2$ into $L_1$
 inserts in an outfix-guided way
 a string of $L_2$ into $L_1$ and then  iteratively
 inserts a string obtained in the process into
$L_1$. The  left one-sided iterated outfix-guided insertion is defined
symmetrically. In fact, when
considering iterated ordinary insertion,
Kari~\cite{Kari91} uses a definition that we  call
left one-sided iterated insertion (and the operation
was defined as $\mathbb{LI}^*(L_1, L_2)$ above).
Csuhaj-Varju et al.~\cite{Csuhaj2007} define iterated overlap
 assembly using  right one-sided
iteration of the operation.

\begin{definition}
\label{datta5}
{\rm 
Let $L_1$ and $L_2$ be languages. The 
{\em right one-sided iterated insertion of $L_2$ into $L_1$\/} is
defined  inductively by setting
$\mathbb{ROGI}^{(0)}(L_1, L_2) = L_2$ and
$\mathbb{ROGI}^{(i+1)}(L_1, L_2) = L_1 \leftarrow
\mathbb{ROGI}^{(i)}(L_1, L_2)$, $i \geq 0$.
The {\em right one-sided insertion closure\/} of $L_2$ into $L_1$ 
is
$
\mathbb{ROGI}^{*}(L_1, L_2) = 
\bigcup_{i=0}^\infty \mathbb{ROGI}^{(i)}(L_1, L_2).
$

 The 
{\em left one-sided iterated insertion of $L_2$ into $L_1$\/} is
defined  inductively by setting
$\mathbb{LOGI}^{(0)}(L_1, L_2) = L_1$ and
$\mathbb{LOGI}^{(i+1)}(L_1, L_2) = 
\mathbb{LOGI}^{(i)}(L_1, L_2) \leftarrow L_2$, $i \geq 0$.
The {\em left one-sided insertion closure\/} of $L_2$ into $L_1$ 
is
$
\mathbb{LOGI}^{*}(L_1, L_2) = 
\bigcup_{i=0}^\infty \mathbb{LOGI}^{(i)}(L_1, L_2).
$
} 
\end{definition}


Note that for any language $L$, 
$
\mathbb{OGI}^{(1)}(L) =
\mathbb{LOGI}^{(1)}(L, L) = 
\mathbb{ROGI}^{(1)}(L, L) = L \leftarrow L.
$

The iterated version of unrestricted
outfix-guided insertion is considerably more general
than the one-sided variants. For any language $L$,
$\mathbb{ROGI}^*(L, L)$ and $\mathbb{LOGI}^*(L, L)$ are
always included in $\mathbb{OGI}^*(L)$ and, in general, the
inclusions can be strict.

\begin{example}
\label{example2}
Let $\Sigma = \{ a, b, c \}$ and $L_1 = \{ aa cc \}$,
$L_2 = \{ a b c \}$.
Now 
$\mathbb{ROGI}^*(L_1, L_2) = a^+ b c^+.$
For example, by inserting $abc$ into
$aa cc$ derives $aabcc$: 
\begin{equation}
\label{eetta19}
a \underline{ac} c \stackrel{[\underline{a}b\underline{c}]}{\Rightarrow}
aabcc.
\end{equation}
A right one-sided iterated insertion of $L_2$ into $L_1$
 could then be continued, for example, as
$
a \underline{acc} \stackrel{[\underline{a}ab\underline{cc}]}{\Rightarrow}
aaabcc$.
In this way right one-sided derivations can generate all strings
of $a^+ b c^+$.  Since all inserted strings must contain the
symbol 
$b$,
the first matched part must always belong to $a^+$ and the
second matched part must belong to $c^+$. This means that
$\mathbb{ROGI}^*(L_1, L_2) \subseteq a^+ b c^+$.

On the other hand,
$
\mathbb{LOGI}^*(L_1, L_2) = \{ aa b cc, aa cc \}.
$
In a left one-sided iterated insertion of $L_2$ into
$L_1$, the only non-trivial   derivation step
is~(\ref{eetta19}).

By denoting $L_3 = L_1 \cup L_2$, it can be verified that
$$
\mathbb{OGI}^*(L_3) = 
\mathbb{ROGI}^*(L_3, L_3) = 
\mathbb{LOGI}^*(L_3, L_3) 
= a^+ b c^+ \cup a^2 a^* c^2 c^*.
$$
\end{example}

The next example illustrates that unrestricted
outfix-guided insertion closure of a language $L'$ can
be larger than $\mathbb{LOGI}^*(L', L')$.
The language $L$ used in the proof of Theorem~\ref{tatta1}
in the next section gives an example where the unrestricted
insertion closure is larger than $\mathbb{ROGI}^*(L, L)$ (as
 explained before Proposition~\ref{latta18}).

\begin{example}
	\label{example3}
Let $ \Sigma = \{ a, b, c, d, e, f \}$
and $L' = \{ abce, bcde, acdef \}$.
We note that
$
a \underline{bce} \stackrel{[\underline{bc} d \underline{e}]}{\Rightarrow}
abcde$.
Furthermore, it is easy to verify
that by outfix-guided inserting strings of $L'$ into $L' \cup \{ abcde \}$
one cannot produce more strings and, thus, 
$\mathbb{LOGI}^*(L', L') = L' \cup \{ abcde \}$.
On the other hand, we have
$$
\underline{acde} f 
\stackrel{ [\underline{a} b \underline{cde}]}{\Rightarrow} abcdef \in
\mathbb{OGI}^{(2)}(L').
$$
\end{example}

%

\if01
$(L_1 \leftarrow L_2)^{(0)} = L_1$\\
 $(L_1 \leftarrow L_2)^{(i+1)} = L_1 \leftarrow  (L_1 \leftarrow L_2)^{(i)}$,
$i \geq 0$

Define
$$
 (L_1 \leftarrow L_2)^{*} = \bigcup_{i=0}^{\infty} (L_1 \leftarrow L_2)^{(i)}
$$

Choose $L_2 = \{ a \$ b \}$, $L_1 = \{ acdb, cabd \}$.

Then
$$
(L_1 \leftarrow L_2)^{*} = \{ (ca)^i \$ (bd)^i \mid i \geq 0 \}
\cup  \{ a(ca)^i \$ (bd)^i b \mid i \geq 0 \}
$$

{\tt The construction relies on the property that all strings that are
inserted will have the marker \$, i.e., strings are inserted into $L_1$
but strings of $L_1$ cannot be inserted into itself. Is it possible
to somehow get rid of this restriction? (i.e., make the construction
work with ``normal'' iterated OGI)}
\fi

\section{Outfix-Guided Insertion and Regular Languages
}

As can be expected, the family of regular languages is
closed under the outfix-guided (prefix-guided, suffix-guided,
respectively) insertion operation.
On the other hand, the answer to the question whether regular languages
are closed under iterated outfix-guided insertion  seems less clear.
From Kari~\cite{Kari91}  we recall
that it is easy to construct examples that establish the
non-closure of regular languages under iterated non-overlapping
insertion. Using variants of such examples we see
that the prefix-guided (or suffix-guided)
insertion closure of a singleton language may be
non-regular. 

On the other hand, analogous straightforward
counter-examples do not work for
the unrestricted outfix-guided insertion closure.
Using a more involved construction we establish that 
the outfix-guided insertion closure of a finite language need
 not be regular. The non-closure of regular languages under
 right one-sided insertion closure is established by a more
 straightforward construction 
 (Proposition~\ref{latta18}).


We begin by showing that regular languages are
closed under non-iterated outfix-guided
insertion. The proof is not surprising but we give an explicit
construction because, essentially, the same construction will be used
to show   in Theorem~\ref{CFREGtheorem}
that the outfix-guided
insertion of a regular (respectively, context-free)
language into a context-free (respectively, regular) language
is always context-free, and for the polynomial time algorithm
to decide whether the language recognized by a DFA is closed under
outfix-guided insertion in section~\ref{kuusi}.

\begin{lemma}
\label{latta20}
If $L_1$ and $L_2$ are regular, then so is $L_1 \leftarrow L_2$.
\end{lemma}

\begin{proof}
Let $L_1$ be recognized by an NFA $A = (\Sigma, Q, \delta, q_0, F_A)$
and $L_2$ be recognized by an NFA $B = (\Sigma, P, \gamma, p_0, F_B)$.
Denote $\overline{Q} =  \{ \overline{q} \mid q \in Q \}$ and
$\overline{P} = \{ \overline{p} \mid p \in P \}$. Here $Q \cup P$
is disjoint with $\overline{Q} \cup \overline{P}$ and $\clubsuit$,
$\heartsuit$ are new symbols not occurring in any of the sets.

For the language $L_1 \leftarrow L_2$ we construct an NFA
$C = (\Sigma, R, \omega, r_0, F_C)$ where
$$
R = Q \times (P \cup \overline{P} \cup \{ \clubsuit, \heartsuit \}) 
\cup \overline{Q} \times P,
$$
$$
F_C = \{ (q, \overline{p}) \mid
q \in F_A, p \in F_B \} \cup \{ (q, \heartsuit) \mid q \in F_A \},
$$
$r_0 = (q_0, \clubsuit)$,
and  for defining the transitions  of $\omega$ 
let $b$ be an arbitrary symbol of $\Sigma$. We set
\begin{enumerate}
\item for $q \in Q$: $\omega((q, \clubsuit), b) = 
\{ (q', \clubsuit) \mid q' \in \delta(q, b) \} \cup
\{ (q', p') \mid q' \in \delta(q, b), p' \in \gamma(p_0, b) \}$,
\item for $q \in Q$,  $p \in P$:
$\omega((q, p), b) = \{ (q', p') \mid q' \in \delta(q, b),
p' \in \gamma(p, b) \} \cup \{ (\overline{q}, p') \mid
p' \in \gamma(p, b) \} \cup \{ (q', \overline{p'}) \mid q' \in \delta(q, b),
p' \in \gamma(p, b) \}$,
\item  for $q \in Q$,  $p \in P$: $\omega((\overline{q}, p), b) 
= \{ \overline{q}, p') \mid
p' \in \gamma(p, b) \} \cup 
\{ (q', \overline{p'}) \mid q' \in \delta(q, b), p' \in \gamma(p, b) \}$,
\item  for $q \in Q$,  $p \in P$: $\omega((q, \overline{p}), b) = 
\{ (q', \overline{p'}) \mid q' \in \delta(q, b), p' \in \gamma(p, b) \} 
\cup Z_p$, where
$$
Z_p = \begin{cases}
\{ (q', \heartsuit) \mid q' \in \delta(q, b) \} & \text{ if }
p \in F_B,\\
\emptyset & \text{ if } p \not\in F_B,
\end{cases}
$$
\item for $q \in Q$: $\omega((q, \heartsuit), b) = 
\{ (q', \heartsuit) \mid q' \in \delta(q, b) \}$.
\end{enumerate}
All transitions not listed above are undefined.

We begin by verifying that $L(A) \leftarrow L(B) \subseteq L(C)$.
Consider a string $w = x_1 u z v x_2$ where $x_1 u v x_2 \in L(A)$
and $uzv \in L(B)$, $u, v \neq \varepsilon$. Roughly speaking,
$C$ uses the states of $Q \times \{ \clubsuit \}$ to process the
prefix $x_1$, the states of $Q \times P$ to process the following
substring $u$, the states of $\overline{Q} \times P$ to process
the substring $z$, the states of $Q \times \overline{P}$ to process
the substring $v$, and the states of $Q \times \{ \heartsuit \}$
to process the suffix $x_2$. Note that according to rules (ii), on
states of $Q \times P$ the NFA simulates $A$ in the first component
and $B$ in the second component of the states. According to
rules (iii), on states of $\overline{Q} \times P$, the NFA $C$ simulates
only $B$ in the second component, and according to rules (iv),
on states of $Q \times \overline{P}$ the NFA $C$ simulates again
both $A$ and $B$ (in the first and second component of the  state of $C$,
respectively).

In more detail, consider an accepting computation 
${\rm comp_A}(x_1 uv x_2)$ of $A$ on $x_1 u v x_2$
that reaches state $q_{x_1}$ (respectively, $q_u$, $q_v$, $q_{x_2}$)
after reading the prefix $x_1$ (respectively, $x_1 u$, 
$x_1 u v$, $x_1 u v x_2$). An accepting computation
of $C$ first reads $x_1$ using rules (i) and simulating 
the  computation ${\rm comp}_A(x_1 uv x_2)$ 
of $A$ on the prefix $x_1$, thus ending
in state $(q_{x_1}, \clubsuit)$. 

When reading the first symbol $b_1$
of $u$, again using a rule (i) the computation of $C$ goes
to a state $(q', p')$ where $q' \in \delta(q_{x_1}, b_1)$ and
the second component begins to simulate an accepting computation
of $B$ on $uzv$ in a state $p' \in \gamma(p_0, b_1)$.
The computation nondeterministically guesses when it sees the first
symbol of $z$, and using rules (ii) enters a state
$(\overline{q_u}, p')$ where $p'$  is the state of $B$ in an accepting
computation on $uzv$ after reading the first symbol of $z$. If $z = 
\varepsilon$, the computation guesses when it sees
the first symbol $b_2$ of $v$ and using the ``third option'' in the rules (ii), 
$C$ goes to a state $(q', p')$ where $q' \in \delta(q_{u}, b_2)$ and
$p'$ is a state that can be reached by $B$ after reading $ub_2$.

The computation processes the substring $z$ in a state
of $\{ \overline{q_u} \} \times P$ using rules (iii) and simulating
the computation of $B$ in the second component. When the
computation guesses that it sees the first symbol $b_2$ of $v$,
using the ``second part'' of the rules (iii) the NFA $C$ goes
to a state  $(q', p')$ where $q' \in \delta(q_{u}, b_2)$ and
$p'$ is a state that can be reached by $B$ after reading the
prefix $uzb_2$. Then, according to rules (iv), $C$ simulates
the computation ${\rm comp}_A(x_1 uv x_2)$
in the first component of the state and $B$ in the second component
of the state. Always when the second component is an element
of $F_B$, according to rules (iv), the computation may enter
a state of the form $(q', \heartsuit)$ that indicates that it has
finished reading the substring $uzv$. 

The remaining computation, using rules (v), simulates 
the computation ${\rm comp}_A(x_1 uv x_2)$ of $A$ on the
first component of the states. The choice of the final states then
guarantees that $C$ accepts in the state $(q_{x_2}, \heartsuit)$.
If $x_2 = \varepsilon$, then the computation of $C$ ends in an
accepting state $(q_{v}, p_f)$ where $p_f \in F_B$ is the state
at the end of the simulated computation of $B$ on $uzv$.

For the converse inclusion we note that the definition of the
transitions of $\omega$ guarantees that any computation of
$C$ ending in an accepting state, must have five parts
$P_1$, $P_2$, $P_3$, $P_4$ and $P_5$, where
 $P_1$  
uses states of $Q \times \{ \clubsuit \}$,  
 $P_2$ uses states of $Q \times P$, 
$P_3$ uses states of $\overline{Q} \times P$, 
$P_4$ uses states of $Q \times \overline{P}$ and
$P_5$ uses states of $Q \times \{ \heartsuit \}$.
The part $P_3$ may be empty if, according to rules (ii), the computation
jumps directly from a state of $P_2$ to a state of $P_4$ and the
part $P_5$ may be empty if the computation $P_4$
ends in an accepting state of the form
$(q, \overline{p})$  ($q \in F_A$, $p \in F_B$).

Since states of $P_1$  and $P_5$ simulate  only a computation of $A$,
states of $P_3$ simulate only a computation of $B$ and states of
$P_2$ and $P_4$ simulate both a computation of $A$ and a computation
of $B$, it is easy to verify that $C$ can have accepting computations
only on strings of $L(A) \leftarrow L(B)$.
\prend
\end{proof}

The result of Lemma~\ref{latta20} extends easily using induction:

\begin{proposition}
\label{patta1}
Suppose $L_1$ and $L_2$ are regular languages.
Then, for all $i \geq 0$, 
$\mathbb{OGI}^{(i)}(L_1)$, $\mathbb{ROGI}^{(i)}(L_1, L_2)$ and 
$\mathbb{LOGI}^{(i)}(L_1, L_2)$ are regular.
\end{proposition}


A simplified variant of the proof of Lemma~\ref{latta20} allows us
to show that the prefix-guided (or suffix-guided) insertion of a regular
language into a regular language is regular. We leave the proof
as an exercise.

\begin{proposition}
\label{patta223}
If $L_1$ and $L_2$ are regular languages, then so are
$L_1 
\stackrel{\rm pgi}{\leftarrow} L_2$ and
$L_1 \stackrel{\rm sgi}{\leftarrow} L_2$.
\end{proposition}

Iterated prefix-guided or suffix-guided insertion does
not preserve regularity.

\begin{proposition}
\label{patta224}
There exist singleton languages $L_1$ and $L_2$ such that
$\mathbb{PGI}^*(L_1)$ and $\mathbb{SGI}^*(L_2)$ are 
non-regular.
\end{proposition}

\begin{proof}
Choose $L_1 = aab$.
We claim that
\begin{equation}
\label{eetta223}
\mathbb{PGI}^*(aab) \cap a^* b^* =
\{ a^i b^j \mid 2 \leq i \leq j+ 1 \}.
\end{equation}
To establish the inclusion from right to left we note
that, for all $i \geq 2$, $a^{i+1} b^{i+2} \in
a^i b^{i+1} \stackrel{\rm pgi}{\leftarrow} aab$.
Furthermore, into strings of the form $a^i b^j$ ($i \geq 2$)
one can always
add exactly one $b$ by inserting $aab$ where prefix $aa$ is
 matched with the last $a$'s of $a^i b^j$.

Second we verify the inclusion from left to right in~(\ref{eetta223}).
Denote $Z = \{ a^i b^j \mid 2 \leq i \leq j+ 1 \}$.
We note that derivations of strings in
$\mathbb{PGI}^*(aab)$ belonging to $a^* b^*$ can use only strings
in $a^* b^*$ because if $w_1 \not\in a^*b^*$ or 
$w_2 \not\in a^* b^*$ then any string in
$w_1 \stackrel{\rm pgi}{\leftarrow} w_2$ has an occurrence
of $b$ preceding an occurence of $a$. It is clear that
for $u_1, u_2 \in Z$, all strings in
$(u_1 \stackrel{\rm pgi}{\leftarrow} u_2) \cap a^* b^*$ 
must be in $Z$. Note that 
 the matched prefix of $u_2$ must contain at least 
one $a$ and thus the insertion adds to $u_1$
at least as many $b$'s as $a$'s.

Since the language $Z$ is non-regular, (\ref{eetta223}) implies
that $\mathbb{PGI}^*(aab)$ is non-regular.

A completely symmetric argument establishes that
$\mathbb{SGI}^*(abb)$ is non-regular.
\prend
\end{proof}

It seems difficult to extend the proof of Lemma~\ref{latta20}
for outfix-guided insertion closure because on strings with
iterated insertions, the computations on corresponding prefix-suffix
pairs can, in general, depend on each other and when processing
a part inserted in between, an NFA would need to keep track
of such pairs, as opposed to simply keep track of a set
of states. On the other hand,  constructions as in the
proof of Proposition~\ref{patta224} rely on the property that the
matched substrings are all either prefixes or all suffixes of the
inserted strings and this type of straightforward
constructions do not yield a regular language whose outfix-guided insertion
closure is non-regular.

Next we show that regular languages, indeed, are not closed under
iterated outfix-guided insertion. For the construction we use 
the following technical
lemma.

\begin{lemma}
\label{latta16}
Let $\Sigma = \{ a_1, a_2, a_3, b_1, b_2, b_3 \}$ and define
$$
L_1 =
\{ a_3 a_1 a_2 b_1, \;\;
a_2 b_2 b_1 b_3, \;\;
a_1 a_2 a_3 b_2, \;\;
a_3 b_3 b_2 b_1, \;\;
a_2 a_3 a_1 b_3, \;\;
a_1 b_1 b_3 b_2
\}.
$$
Then $L_1 \leftarrow L_1 = L_1$.
\end{lemma}

\begin{proof}
The inclusion from right to left follows from the observation
that any string $w$ of length at least two is a non-trivial
outfix of itself and, consequently $w$ can be inserted into itself
 to produce $w$
as a result.

For the converse inclusion we verify that for all $x, y \in L_1$, 
if $y \neq x$, then $y$ cannot be outfix-guided inserted into $x$
and $x$ can be outfix-guided inserted into itself only in the trivial
way of using as matching parts a non-empty prefix $x_1$
and a non-empty suffix $x_2$ such that $x_1 x_2 = x$. 
The second claim is obvious
because each string of $L_1$ consists of 4 different symbols.

Consider now $x, y \in L_1$, $x \neq y$. 
For the sake of contradiction suppose that
$w \in x \leftarrow y$ and that $w$ is obtained by matching
substrings $u$ and $v$ of $x$ with a prefix and a suffix of $y$,
respectively. The substrings $u$ and $v$ cannot both consist
of symbols $a_i$, $1 \leq i \leq 3$, because if $x$ and $y$ both
contain more than one symbol $a_i$, they must end with symbols
$b_{j_1}$ and $b_{j_2}$, $j_1 \neq j_2$. Using a symmetric argument
we observe that $u$ and $v$ cannot both consist of symbols
$b_i$, $1 \leq i \leq 3$.

The remaining possibility is that the first matched part $u$ consists
of (one or more)  symbols $a_i$ and the second matched part $v$ consists of
(one or more) symbol $b_j$. For simplicity in the following discussion
we assume that $v$ begins with $b_1$. The definition of $L_1$ is symmetric,
and an analogous argument works when the first symbol of $v$ is
$b_2$ or $b_3$.

Now $uv$ must be a substring of $x$. This means that the last symbol
of $u$ can be $a_2$ if $x = a_3 a_1 a_2 b_1$ or $a_1$ if
$x = a_1 b_1 b_3 b_2$. Besides the string $ a_3 a_1 a_2 b_1$
the only other string of $L_1$ where $a_2$ occurs before $b_1$
is $a_2 b_2 b_1 b_3$. The string  $a_2 b_2 b_1 b_3$
cannot be inserted into  $x =  a_3 a_1 a_2 b_1$ because the last
symbol $b_3$ would be ``outside'' of $x$.

As the remaining case consider then the possibility
$x = a_1 b_1 b_3 b_2$. The only string of $L_1 - \{ x \}$ where
$a_1$ occurs before $b_1$ is $a_3 a_1 a_2 b_1$ and again this cannot
be inserted into $x$ because the first symbol $a_3$ would be
outside of $x$.
\prend
\end{proof}

\begin{theorem}
\label{tatta1}
There exists a finite language $L$ such that
$\mathbb{OGI}^*(L)$ is non-regular.
\end{theorem}

\begin{proof} 
Let $\Sigma = \{ a_1, a_2, a_3, b_1, b_2, b_3  \}$ and define
$L \subseteq (\Sigma \cup \{ \$ \})^*$ as
$$
L = \{  \$ a_3 a_1 b_1 b_3 \$, \;
a_3 a_1 a_2 b_1, \;
a_2 b_2 b_1 b_3, \;
a_1 a_2 a_3 b_2, \;
a_3 b_3 b_2 b_1, \;
a_2 a_3 a_1 b_3, \;
a_1 b_1 b_3 b_2
\}.
$$
Note that $L - \{ \$ a_3 a_1 b_1 b_3 \$ \}$ is equal to the language 
$L_1$ from Lemma~\ref{latta16}.
Our construction is based on an idea that the only way to produce
new strings in $\mathbb{OGI}^*(L)$ is to insert into strings
obtained from $  \$ a_3 a_1 b_1 b_3 \$ $ cyclically copies of the
strings of $L_1$. 
For ease of discussion we introduce 
names for the strings  of $L_1$:
\begin{eqnarray*}
y_1 = a_3 a_1 a_2 b_1, \;\;
y_2 = a_2 b_2 b_1 b_3, \;\;
y_3 = a_1 a_2 a_3 b_2, \;\;
y_4 = a_3 b_3 b_2 b_1, \\
y_5 = a_2 a_3 a_1 b_3, \;\;
y_6 = a_1 b_1 b_3 b_2.
\end{eqnarray*}
For specifying the language $\mathbb{OGI}^*(L)$ we define
the finite set 
$$
S_{\rm middle} = \{   a_1 b_1, \, a_1 a_2 b_1,  
a_1 a_2 b_2 b_1, \,
a_1 a_2 a_3 b_2 b_1, \,
a_1 a_2 a_3 b_3 b_2 b_1, \,
a_1 a_2 a_3 a_1 b_3 b_2 b_1  \}. 
$$
We claim that
\begin{equation}
\label{eetta1}
\mathbb{OGI}^*(L) = 
\{ \$ a_3 (a_1 a_2 a_3)^i z (b_3 b_2 b_1)^i b_3 \$ \mid i \geq 0,\;
z \in S_{\rm middle} \}.
\end{equation}

To establish the inclusion from right to left, we note that 
\begin{eqnarray*}
\$ a_3 a_1 b_1 b_3 \$ \stackrel{[y_1]}{\Rightarrow}
\$ a_3 a_1 a_2 b_1 b_3 \$ \stackrel{[y_2]}{\Rightarrow}
\$ a_3 a_1 a_2 b_2 b_1 b_3 \$ \stackrel{[y_3]}{\Rightarrow} 
\$ a_3 a_1 a_2 a_3 b_2 b_1 b_3 \$ \stackrel{[y_4]}{\Rightarrow}  \\
\$ a_3 a_1 a_2 a_3 b_3 b_2 b_1 b_3 \$ \stackrel{[y_5]}{\Rightarrow}
\$ a_3 a_1 a_2 a_3 a_1 b_3 b_2 b_1 b_3 \$ \stackrel{[y_6]}{\Rightarrow}
\$ a_3 a_1 a_2 a_3 a_1 b_1 b_3 b_2 b_1 b_3 \$ = w_1 .
\end{eqnarray*} 
The first five insertions  generate the strings
$ \$ a_3 z b_3 \$ $, $z \in S_{\rm middle}$, and the last
string $w_1$ again has ``middle part'' $a_3 a_1 b_1 b_3$. 
By cyclically outfix-guided inserting the strings $y_1, \ldots, y_6$
into $w_1$ we get all strings
$ \$ a_3 (a_1 a_2 a_3) z (b_3 b_2 b_1) b_3 \$ $, 
$z \in S_{\rm middle}$, and the string
$ \$ a_3 (a_1 a_2 a_3)^2 a_1 b_1 (b_3 b_2 b_1)^2 b_3 \$ $. By  simple
induction it follows that $\mathbb{OGI}^*(L)$ contains the
right side of~(\ref{eetta1}).

To establish the converse inclusion, we verify that all strings
obtained by iterated outfix-guided insertion from strings of $L$
must be obtained as above, that is, all non-trivial derivations 
producing new strings must be as above. 

Since $ \$ a_3 a_1 b_1 b_3 \$ $ is the only string in $L$
containing symbols \$ and they occur as the first and the last
symbol, it is clear that all strings in $\mathbb{OGI}^*(L)$
containing symbols \$ must be in $ \$ \Sigma^* \$ $.
A string of  $ \$ \Sigma^* \$ $ cannot be outfix-guided inserted
to any string not containing symbols \$ and a string of
 $ \$ \Sigma^* \$ $ can be outfix-guided inserted into
another string of  $ \$ \Sigma^* \$ $ only using a trivial
derivation step.

By Lemma~\ref{latta16} we know that strings of $L_1$ cannot be
outfix-guided inserted into other strings of $L_1$.

We have verified that the set
$$
L_{\rm gen} = \{ \$ a_3 (a_1 a_2 a_3)^i z (b_3 b_2 b_1)^i b_3 \$
\mid i \geq 0,\;
z \in S_{\rm middle} \}
$$
is included in $\mathbb{OGI}^*(L)$ and strings of
$L_{\rm gen}$ can be inserted into strings of
$L_{\rm gen}$ only in a trivial way. To complete the proof it
remains to verify that inserting strings of $L_1$ into
$L_{\rm gen}$ does not produce additional strings, that is,
$L_{\rm gen} \leftarrow L_1 \subseteq L_{\rm gen}$.

It is impossible  to insert $ a_3 a_1 a_2 b_1$ into a string of
$L_{\rm gen}$ using an outfix obtained from $a_3$ and $a_2 b_1$
because in strings of $L_{\rm gen}$ the symbols $a_3$ and $a_2$
do not occur consecutively. The same applies to the other five
strings $y_2, \ldots, y_6 \in L_1$: we cannot insert $y_j$ into
$L_{\rm gen}$ using an outfix where the prefix ends and the
 suffix begins
 with a symbol of $\{ a_1, a_2, a_3 \}$ (respectively, with
a symbol of $\{ b_1, b_2, b_3 \}$).

The other non-trivial possibilities are that we insert
$a_3 a_1 a_2 b_1$ into a string of $L_{\rm gen}$ using
an outfix $uv$ where $u$ is either $a_3$ or $a_3 a_1$ and
$v = b_1$. The choice $u = a_3$, $v = b_1$ is not possible because
$a_3 b_1$ is not a substring of a string in $L_{\rm gen}$.
The insertion using outfix $a_3 a_1 b_1$ can be done only
to a string of the from
$ \$ a_3 (a_1 a_2 a_3)^i a_1 b_1 (b_3 b_2 b_1)^i b_3 \$ $, $i \geq 0$
and it produces
$ \$ a_3 (a_1 a_2 a_3)^i a_1 a_2 b_1 (b_3 b_2 b_1)^i b_3 \$ 
\in L_{\rm gen}$.
Using symmetry of the definition of $L_1$, the argument for the
other five strings  of $L_1$ is completely analogous.

This establishes~(\ref{eetta1}) and the non-regularity
of $\mathbb{OGI}^*(L)$.
\prend
\end{proof}


We conjecture that the iterated outfix-guided insertion closure
of a regular language need not be even context-free.
However, a construction of such a language would seem to be
considerably more complicated than the construction used
in the proof of Theorem~\ref{tatta1}.

\begin{open}
\label{oatta1}
Find a regular (or a finite) language $L$ such that $\mathbb{OGI}^*(L)$
is not context-free.
\end{open}

Contrasting the result of Theorem~\ref{tatta1} we show that unary
regular languages are closed under iterated outfix-guided
insertion. The construction is based on a
technical lemma which shows that, for unary languages, outfix-guided
insertion closure can be represented as a variant of the iterated
overlap assembly~\cite{Csuhaj2007,EnagantiIKK15}.

\begin{definition}
\label{doverlap}
{\rm 
Let $x, y \in \Sigma^*$. The {\em 2-overlap catenation\/} of $x$ and
$y$, $x \overline{\odot}^2 y$,
is defined as
$$
x \overline{\odot}^2 y = 
 \{ z \in \Sigma^+ \mid
(\exists u, w \in \Sigma^*)(\exists v \in \Sigma^{\geq 2}) \; 
x = uv, y = vw, z = uvw \}.
$$
2-overlap catenation is extended in the natural to an
operation on languages.
For $L \subseteq \Sigma^*$, we define inductively
$
2\mathbb{OC}^{(0)}(L) = L$  and 
$2\mathbb{OC}^{(i+1)}(L) =$ \\
$2\mathbb{OC}^{(i)}(L) \overline{\odot}^2
2\mathbb{OC}^{(i)}(L)$, $i \geq 0$.
The {\em 2-overlap catenation closure\/} of $L$ is
$
2\mathbb{OC}^*(L) = \bigcup_{i=0}^{\infty} 2\mathbb{OC}^{(i)}(L).
$
}
\end{definition}

Due to commutativity of unary languages we get the following
property which will be crucial for establishing closure of unary
regular languages under outfix-guided insertion closure.

\begin{lemma}
\label{latta23}
If $x, y \in  a^*$ are unary strings, then
$
x \leftarrow y = x \overline{\odot}^2 y.
$
\end{lemma}

\begin{proof}
Consider $w \in (x \leftarrow y)$,
that is, we can write $w = x_1 u z v x_2$, where
$x = x_1 u v x_2$, $y = u z v$ and $u, v \neq \varepsilon$.
Since concatenation of unary strings is commutative, we have
$$
w = x_1 x_2 u v z, \mbox{ where } x = x_1 x_2 uv , \;\; y = uvz.
$$
This establishes that $w \in x \overline{\odot}^2 y$. 

Conversely, 
consider  $z \in x \overline{\odot}^2 y$, that is, $z = uvw$ where
$x = uv$, $y = vw$ and $| v | \geq 2$.
Write $v = v_1 v_2$ where $v_1, v_2 \in \Sigma^+$.
Now, again relying just on commutativity of
unary concatenation, $z = u v_1 w  v_2 \in x \leftarrow y$.
\prend
\end{proof}

\begin{corollary}
\label{catta24}
If $L$ is a unary language then
$\mathbb{OGI}^*(L) = 2\mathbb{OC}^*(L)$.
\end{corollary}

The 2-overlap closure of a regular language is always
regular. The construction does not depend on a language being
unary, so we state the result for regular languages over an
arbitrary alphabet. 
Csuhaj-Varju et al.~\cite{Csuhaj2007} have shown that
iterated overlap 
 assembly preserves regularity. 
The proof of Lemma~\ref{latta25}
is inspired by Theorem~4 of~\cite{Csuhaj2007} but  does not  follow
from it  because
\cite{Csuhaj2007}  defines iteration of operations as right one-sided 
iteration and,
furthermore, 2-overlap catenation  has an additional
length restriction on the overlapping strings.

\begin{lemma}
\label{latta25}
The 2-overlap catenation closure of a regular language is regular.
\end{lemma}

\begin{proof}
Consider a regular language $L$ recognized by an
NFA $A = (\Sigma, Q, \delta, q_0, F)$.
We construct for the language
$2\mathbb{OC}^*(L)$ an NFA $B = 
(\Sigma, 2^{Q}, \gamma, \{ q_0 \}, 
2^F - \{ \emptyset \} )$
where the transitions of $\gamma$ are defined below.


For $\emptyset \neq P \subseteq Q $ and  
$b \in \Sigma$ define
\begin{eqnarray*}
\gamma(P, b) & =  \{ & \; (\delta(P, b) - X_{\rm rem}) \cup
Y_{\rm add} \; \mid \; X_{\rm rem} \subseteq F, \;
\delta(P, b) - X_{\rm rem}  \neq \emptyset, \\
& & 
Y_{\rm add} \subseteq \delta(q_0, b) \;\; \}.
\end{eqnarray*}

States of $B$ are subsets
of $Q$ and the transition relation is nondeterministic:
$\gamma(P, b)$ is a collection of subsets of $Q$.

The computation of $B$  simulates multiple
computations of $A$. When reading a symbol $b \in \Sigma$,
the NFA $B$ can guess that this occurrence of $b$ begins
(one or more) 2-overlap-catenated strings, and adds to the
simulated computations the corresponding states
of $\delta(q_0, b)$. 
Always when a simulated computation
reaches a state of $F$, the NFA $B$ can nondeterministically
guess that $b$ ends a string that is 2-overlap concatenated
with another string.  This is done by the choice of
the set $X_{\rm rem} \subseteq F$ in the definition of $\gamma$.
Note that the
condition $\delta(P, b) - X_{\rm rem}  \neq \emptyset$ 
guarantees that at least one of the simulated computations
that were originated before reading $b$ must remain alive: this
enforces that the overlap with the new computations indeed will
be at least two.

It is clear that, by always choosing the sets $X_{\rm rem}$ and
$Y_{\rm add}$ correctly, $B$ has a computation on an
arbitrary string
in $w \in 2\mathbb{OC}^*(L(A))$ 
that ends in a state $P \subseteq F$,
$P \neq \emptyset$. The set $P$ consists of final states of $A$
that appear in an accepting computation in all strings that in
the representation of $w$ as a
 2-overlap catenation of strings of $L(A)$ are a suffix of $w$.
Note that the construction works also if $w$ has length one: in this case
$w$ must be an element of $L(A)$.

To verify the converse inclusion $L(B) \subseteq
2\mathbb{OC}^*(L(A))$, we note that, in general, some parts
of computations
of $B$ need not simulate any iterated 2-overlap concatenation of
a set of strings of $L(A)$. For example, if $A$
accepts both $x y z$ and $y$ where $|y| \geq 2$,  on the
string $x y z$ the NFA $B$ can begin a second computation $C_2$
when reading the first symbol of $y$ and this
computation then may end in a
final state at the end of the substring $y$.
However, the existence of
the superfluous computation $C_2$ cannot lead to new illegal computations
because the transitions of $\gamma$ add new computations
depending only on the input symbol and not  on the current state.
(In the transitions of $\gamma$, the sets $Y_{\rm add}$ depend only
on the input symbol and the initial state of $A$.)
Thus, the added superfluous computations cannot cause $B$ to accept 
strings not in $2\mathbb{OC}^*(L(A))$.
\prend
\end{proof}

By Corollary~\ref{catta24} and Lemma~\ref{latta25}
we have shown that unary regular languages are closed
under outfix-guided insertion closure, constrasting the
result of Theorem~\ref{tatta1} for general regular languages.

\begin{theorem}
\label{tatta22}
The outfix-guided insertion closure of a unary regular
language is always regular.
\end{theorem}

\subsection{One-sided iterated outfix-guided insertion}

The left and right one-sided insertion closures are restricted
variants of the general outfix-guided insertion closure, so 
Theorem~\ref{tatta1} does not directly imply
 the existence of regular languages $L_1$ and $L_2$
such that $\mathbb{LOGI}^*(L_1, L_2)$ or $\mathbb{ROGI}^*(L_1, L_2)$
are non-regular.
Here we show that the  one-sided outfix-guided insertion
closures are not, in general, regularity preserving. 
For the left-one one-sided outfix-guided insertion closure
the construction is similar to that used in the proof
of 	Theorem~\ref{tatta1}. However, this construction does
not work for right one-sided closure because
if  $L$ is the language
used in the proof of Theorem~\ref{tatta1}, then
$\mathbb{ROGI}^*(L, L)$ is the finite language
$L \cup \{ \$ a_3 a_1 a_2 b_1 b_3 \$ \}$.

\begin{proposition}
\label{latta18}
There exist finite languages $L_1$, $L_2$, $L_3$ and
$L_4$ such that\\
$\mathbb{ROGI}^{*}(L_1, L_2)$ 
and $\mathbb{LOGI}^{*}(L_3, L_4)$ are non-regular.
\end{proposition}

\begin{proof}
We consider first the right one-sided outfix-guided insertion closure.
Let $\Sigma = \{ \$, a, b, c, d \}$ and
choose $L_2 = \{ a \$ b \}$, $L_1 = \{ acdb, cabd \}$.
Then
$$
\mathbb{ROGI}^*(L_1, L_2) = \{ (ca)^i \$ (bd)^i \mid i \geq 0 \}
\cup  \{ a(ca)^i \$ (bd)^i b \mid i \geq 0 \},
$$
which is non-regular. Inserting $a \$ b$ into $cabd$ derives
$ca \$ bd$ and next, in  a right one-sided derivation,
inserting the latter string  into $acdb$ derives
$aca \$ bdb$. Continuing in this way we get a right one-sided derivation
for all strings in the set appearing on the right side of the
equation.

The fact that $\mathbb{ROGI}^*(L_1, L_2)$ does not
contain any additional strings follows
from the property of right one-sided iterated insertions:
all strings that are inserted into $L_1$ will have the marker
\$ and strings of $L_1$ cannot be inserted into strings of $L_1$.
We leave to the reader
the details of verifying that the inclusion holds from
left to right.

The construction of the languages $L_3$ and $L_4$ for the
left one-sided outfix-guided insertion closure is obtained
by modifying the language in the proof of Theorem~\ref{tatta1}.
Let $\Sigma = \{ a_1, a_2, a_3, b_1, b_2, b_3  \}$ and define
$L_3 = \{  \$ a_3 a_1 b_1 b_3 \$ \}$ and 
$$
L_4 = \{  
a_3 a_1 a_2 b_1, \;\;
a_2 b_2 b_1 b_3, \;\;
a_1 a_2 a_3 b_2, \;\;
a_3 b_3 b_2 b_1, \;\;
a_2 a_3 a_1 b_3, \;\;
a_1 b_1 b_3 b_2
\}.
$$
Denote
$$
L_5 = \{ \$ a_3 (a_1 a_2 a_3)^i z (b_3 b_2 b_1)^i b_3 \$ \mid i \geq 0,\;
z \in S_{\rm middle} \},
$$
where
$S_{\rm middle} = \{   a_1 b_1, \; a_1 a_2 b_1,  \;
a_1 a_2 b_2 b_1, \;
a_1 a_2 a_3 b_2 b_1, \;
a_1 a_2 a_3 b_3 b_2 b_1, \;
a_1 a_2 a_3 a_1 b_3 b_2 b_1  \}$.

From the proof of Theorem~\ref{tatta1} it follows that
\begin{equation*}
\mathbb{LOGI}^*(L_3, L_4) =  L_5.
\end{equation*}
Note that the first part of the proof of 
Theorem~\ref{tatta1} establishes that all strings of
$L_5$ 
are obtained by left one-sided iterated insertion of $L_4$
into $ \$ a_3 a_1 b_1 b_3 \$ $. Thus, 
$L_5 \subseteq \mathbb{LOGI}^*(L_3, L_4)$. The proof of
Theorem~\ref{tatta1} also establishes that
$\mathbb{OGI}^*(L_3 \cup L_4) = L_5$ and directly by the definition
of the iterated operations,
$\mathbb{LOGI}^*(L_3, L_4) \subseteq \mathbb{OGI}^*(L_3 \cup L_4)$.
\prend
\end{proof}


%

\section{Outfix-Guided Insertion and Context-Free Languages}
\label{sec-context-free}

It is well known that the family of context-free languages
is closed under ordinary insertion.
We show that context-free languages are not closed under
outfix-guided (or prefix-guided, suffix-guided, respectively) insertion.
This contrasts also  the corresponding result 
for regular languages from Lemma~\ref{latta20}. 

\begin{theorem}
\label{CFtheorem}
There exists a context-free language $L$ such that
$L \leftarrow L$
is not context-free.
\end{theorem}

\begin{proof}
Let $\Sigma = \{ \$, a, b, c \}$. By choosing
$$
L = \{ \$ a^n \$ \$ c^n \mid n \geq 1 \} \cup
\{ \$ a^n \$ b^n \$ \mid n \geq 1 \}
$$
we note that
$$
(L \leftarrow L) \cap \$ a^+ \$ b^+ \$ c^+
= \{ \$ a^n \$ b^n \$ c^n \mid n \geq 1 \}. 
$$
The claim follows since the intersection of a context-free
language and a regular language is always
context-free~\cite{Shallit09}.
\prend
\end{proof}

The same language $L$ as in the proof of Theorem~\ref{CFtheorem}
can be used to establish that context-free languages are not closed
under prefix-guided insertion and the reversal of $L$ can be used
to establish non-closure under suffix-guided insertion.

\begin{corollary}
\label{catta223}
There exist context-free languages $L_1$ and $L_2$ such that
$L_1 \stackrel{\rm pgi}{\leftarrow} L_1$ and
$L_2 \stackrel{\rm sgi}{\leftarrow} L_2$ are not context-free.
\end{corollary}

On the other hand, the outfix-guided insertion of a regular
(respectively, context-free)
language into a context-free (respectively, regular) language
is always context-free.

\begin{theorem}
\label{CFREGtheorem}
If $L_1$ is context-free and $L_2$ is regular, then
$L_1 \leftarrow L_2$ and $L_2 \leftarrow L_1$ are context-free.
\end{theorem}

\begin{proof}
Suppose $L_1$ is recognized by a nondeterministic PDA
$M$ and $L_2$ 		is recognized by an NFA $A$. By combining
the finite state transitions of $M$ and $A$ as in the proof of
Lemma~\ref{latta20}, and simultaneously simulating the pushdown stack
of $M$ we can construct a PDA $M_1$ for $L_1 \leftarrow L_2$.
Always when $M_1$ makes a transition simulating a transition
of $M$, it makes a corresponding stack operation. On
the other hand, transitions
of $M_1$ simulating only transitions of $A$ do not touch the stack.
A PDA for $L_2 \leftarrow L_1$ is obtained by interchanging in the
construction of Lemma~\ref{latta20} the roles of $M$ and $A$.
\prend
\end{proof}

The analogy of Theorem~\ref{CFREGtheorem} does not hold for
deterministic context-free languages.
Techniques for proving that a language is not deterministic
context-free are known already from \cite{Ginsburg66}.



\begin{theorem}
\label{tatta33}
If $L_1$ is deterministic context-free and $L_2$ is regular, the languages
$L_1 \leftarrow L_2$ or $L_2 \leftarrow L_1$ need not be deterministic
context-free.
\end{theorem}

\begin{proof}
First we show that
there exist a DCFL $L_1$ and a regular language $L_2$ such 
that $L_1 \leftarrow L_2$ is not deterministic context-free.

Let $L_1 = \{c d a^i b^i a^j \mid i,j \geq 1\} 
\cup \{c a^i b^j a^j \mid i,j \geq 1\}$ and $L_2 = \{cda\}$. 
We can outfix-guided insert $cda$ in a non-trivial way only
into words of the form $c a^i b^j a^j$, which gives us
\begin{equation*}
L_1 \leftarrow L_2 = cd \cdot (\{a^i b^i a^j \mid i,j \geq 1\} \cup \{a^i b^j a^j \mid i,j \geq 1\}).
\end{equation*}
From \cite{Ginsburg66}, we have that 
$L \subseteq (\Sigma - c)^*$ is a DCFL if and 
only if $cL$ is a DCFL and that the language 
$(\{a^i b^i a^j \mid i,j \geq 1\} \cup \{a^i b^j a^j \mid i,j \geq 1\})$ 
is not deterministic. Thus, $L_1 \leftarrow L_2$ is not a DCFL.

Second, we show that
there exist a regular language $L_3$ and a DCFL $L_4$ such that 
$L_3 \leftarrow L_4$ is not deterministic context-free.

Let $L_3 = (a^* b a c) + (aba^*)$ and 
$L_4 = \{b^j a^j c \mid j \geq 1\} \cup \{a^i b^i a^2 \mid i \geq 1\}$. 
We can  insert words of the form $b^j a^j c$ in a non-trivial
way only into words of the form 
$a^i b a c$. This gives us the set $\{ a^i b^j a^j c \mid i,j \geq 1 \}$. 
Similarly, we can only insert words of the form $a^i b^i a^2$ 
into words of the form $a b a^2 a^j$, resulting in the set 
$\{ a^i b^i a^j \mid i \geq 1, j \geq 2 \}$. Thus,
\begin{equation*}
L_3 \leftarrow L_4 = \{ a^i b^j a^j c \mid i,j \geq 1 \} 
\cup \{ a^i b^i a^j \mid i \geq 1, j \geq 2 \},
\end{equation*}
which is not deterministic context-free.
\prend
\end{proof}

Theorem~\ref{CFtheorem} raises the question how complex languages can
be obtained from context-free languages using iterated outfix-guided
insertion. Note that if $L_1$ and $L_2$ are context-free, it is
easy to verify that $L_1 \leftarrow L_2$ is  deterministic
context-sensitive. Next we consider the corresponding question
for the insertion closures.

\begin{proposition}
\label{patta31}
If $L_1$ and $L_2$ are context-free then
$\mathbb{ROGI}^*(L_1, L_2)$ and\\ $\mathbb{LOGI}^*(L_1, L_2)$
are context-sensitive.
\end{proposition}

\begin{proof}
We consider only the right one-sided insertion closure -- the
proof for left one-sided insertion closure is similar.

Below by a substring occurrence of $w$ we mean a 
unique substring beginning
at a specified position in $w$.
From the definition of right one-sided iterated insertion it
follows that $w \in \mathbb{ROGI}^*(L_1, L_2)$ if and only
if there exists $k \geq 1$ and a sequence of substring 
occurrences of $w$:  $s_1, s_2, \ldots, s_k$ where $s_k = w$ and
$s_i$ is always inside the substring occurrence $s_{i+1}$,
$i = 1, \ldots, k-1$, and:
\begin{itemize}
\item
We can write $s_1 = x_1 u z v x_2$, where $x_1 u v x_2 \in L_1$,
$u z v \in L_2$.
\item We can write $s_2 = x_1' u' z' v' x_2'$, where
$x_1' u' v' x_2' \in L_1$, $u' z' v' = s_1$, 
\item \ldots
\item We can write $s_k = x_1'' u'' z'' v'' x_2''$, where
$x_1'' u'' v'' x_2'' \in L_1$, $u'' z'' v'' = s_{k-1}$.
\end{itemize}
Furthermore, we can assume that $|s_i| < |s_{i+1}|$, $i = 1, \ldots, k-1$,
because if this is not the case, in the chain we can simply omit
$s_{i+1}$. Now, on input $w$, a nondeterministic linear space Turing machine
$M$ can begin by guessing $s_1$ and verifying that it has the required
decomposition. In the $(i+1)$st stage $M$ always ``remembers''
(by markers on the tape) the previous string $s_i$, then guesses 
the substring $s_{i+1}$ (where $s_i$ is a substring of $s_{i+1}$)
and verifies that the conditions hold for $s_{i+1}$. At the end $M$
accepts if $s_k = w$. Since the values $|s_i|$ form a strictly
increasing sequence, the process can be ended after at most $|w|$
stages.
\prend
\end{proof}

In the  proof of Proposition~\ref{patta31} it is sufficient to know
that the languages $L_1$ and $L_2$ are context-sensitive, 
and as a consequence it follows that context-sensitive languages
are closed under one-sided outfix-guided insertion closure.


We conjecture that, for any context-free language $L$,
 $\mathbb{OGI}^*(L)$ must be context-sensitive.
Constructing a linear bounded automaton for $\mathbb{OGI}^*(L)$
is more difficult than in the case of the right or left one-sided
insertion closures, because a direct simulation of a derivation
of $w \in \mathbb{OGI}^*(L)$ 
(i.e., simulation of the iterated outfix-guided insertion
steps producing $w$)
would need to remember, at a 
given time, an unbounded
number of substrings of the input.

Also we do not know how to make the procedure in the proof of
Proposition~\ref{patta31} deterministic and it remains open
whether the one-sided outfix-guided insertion closures of context-free
languages are always deterministic context-sensitive.

\section{Deciding Closure  under Outfix-Guided Insertion
}
\label{kuusi}

In this section we consider the question whether a
given  language is closed under outfix-guided insertion and show
that  this question is undecidable for context-free languages.

We say that a language $L$ is {\em closed\/} under outfix-guided insertion,
or {\em og-closed\/} for short,
if outfix-guided
inserting strings of $L$ into $L$ does not produce strings outside
of $L$, that is, $(L \leftarrow L) \subseteq L$.

A natural algorithmic problem is then to decide for a given language $L$
whether or not $L$ is og-closed.
If $L$ is regular, by Lemma~\ref{latta20}, we can decide
whether or not $L$ is og-closed.  
For a given DFA $A$, Lemma~\ref{latta20} yields only an
NFA for the language $L(A) \leftarrow L(A)$. In general,
the  NFA equivalence or inclusion problem
is PSPACE complete~\cite{Yu97}. However, inclusion of an NFA
language in the
language $L(A)$ can be tested efficiently when $A$ is deterministic.

\begin{theorem}
\label{patta51}
There is a polynomial time algorithm to decide whether for a
given DFA $A$ the language $L(A)$ is og-closed.
\end{theorem}

\begin{proof}
As in the proof of Lemma~\ref{latta20} we construct an NFA
$B$ for the language $L(A) \leftarrow L(A)$. The number of
states of $B$ is quadratic in the number of states of $A$.
Let $A'$ be the DFA obtained from $A$ by interchanging the final
states and the non-final states.
Now $L(B) \subseteq L(A)$ if and only if $L(B) \cap L(A') = 
\emptyset$ and intersection emptiness for NFAs can be tested
in polynomial time.
\prend
\end{proof}

The method used in Theorem~\ref{patta51} does not yield
an efficient algorithm if the regular language $L$ is specified
by an NFA. The complexity of deciding og-closure of a language
accepted by an NFA
 remains open.
On the other hand, using a reduction from the Post Correspondence
Problem it follows that the question whether or not
a context-free language
is og-closed in undecidable.

\begin{theorem}
\label{tatta52}
For a  context-free language $L$, specified e.g. by a context-free
grammar, the question whether
or not $L$ is og-closed is undecidable.
\end{theorem}

\begin{proof}
Recall that an
instance of the {\em Post Correspondence Problem\/} (PCP)~\cite{Shallit09}
 consists of two lists
of strings $((u_1, \ldots, u_n), (v_1, \ldots, v_n))$,
$u_i, v_i \in \Sigma^*$, $1 \leq i \leq n$, and a solution of this
instance is a sequence of integers
$(i_1, \ldots, i_k)$, $i_j \in \{ 1, \ldots, n \}$, 
$j = 1, \ldots, k$, such that
$u_{i_1} \cdots u_{i_k} = v_{i_1} \cdots  v_{i_k}$. It is well
known that deciding whether or not a
PCP instance has a solution is undecidable~\cite{Shallit09}.

Let $I_{\rm PCP} = ((u_1, \ldots, u_n), (v_1, \ldots, v_n))$, $u_i, v_i 
\in \{ a, b \}^*$, $1 \leq i \leq n$ be an arbitrary
instance of PCP. Choose
$\Sigma = \{ a, b, f, \cent, \$, \# \}$ and define
\begin{eqnarray*}
L_1 & =  \{ \; & \cent \$ i_1 i_2 \cdots i_r \# u_{i_r}^R u_{i_{r-1}}^R \cdots
u_{i_1}^R \# \# v_{j_1} v_{j_2} \cdots v_{j_s} \# j_s j_{s-1} \cdots
j_1 \$ \cent \; \mid \; \\
& &  r, s \geq 1, 1 \leq i_x, j_y \leq n, 1 \leq x \leq r, 1 \leq y \leq s
\; \}, \; \mbox{ and, }
\end{eqnarray*}
\begin{eqnarray*}
L_2 & = \{ \; & \$ i_1 i_2 \cdots i_r \# w \# f \# w^R \# i_r
i_{r-1} \cdots i_1 \$ \; \mid  \;  \\
& &  w \in \{ a, b \}^*, r \geq 1, 1 \leq i_x \leq n, 1 \leq x \leq r \; \}.
\end{eqnarray*}
The languages $L_1$ and $L_2$ are context-free. (The
language $L_2$ can be generated
by a linear context-free grammar and $L_1$ is the concatenation
of two linear context-free languages.)

We define $L = L_1 \cup L_2$ and claim that 
the instance $I_{\rm PCP}$ has a solution if and only if
$L$ is not og-closed.
Below we prove both implications  of the claim.
\begin{description}
\item[{\rm $\bullet$ ``$I_{\rm PCP}$ has a solution
implies $L$ is not closed'':}] Suppose that
$(i_1, \ldots, i_k)$ is a solution for $I_{\rm PCP}$. 
Now
$$
w_1 = \cent \$ i_1 i_2 \cdots i_k \# u_{i_k}^R u_{i_{k-1}}^R \cdots
u_{i_1}^R \# \# v_{i_1} v_{i_2} \cdots v_{i_k} \# i_k i_{k-1} \cdots
i_1 \$ \cent  \in L_1 \; (\subseteq L).
$$
Also since $(i_1, \ldots, i_k)$ is solution we note that
$$
w_2 = \$ i_1 i_2 \cdots i_k \# u_{i_k}^R u_{i_{k-1}}^R \cdots
u_{i_1}^R \# f \# v_{i_1} v_{i_2} \cdots v_{i_k} \# i_k
i_{k-1} \cdots i_1 \$ \in L_2 \; (\subseteq L).
$$
As illustrated in Fig.~\ref{fatta1},
the string $w_2$ can be (in a unique way) outfix-guided inserted
into the string $w_1$ and the resulting string is not in $L$
because no string of $L$ contains both symbols $\cent$ and $f$.
\begin{figure}[htb]
\includegraphics[width=12.5cm]{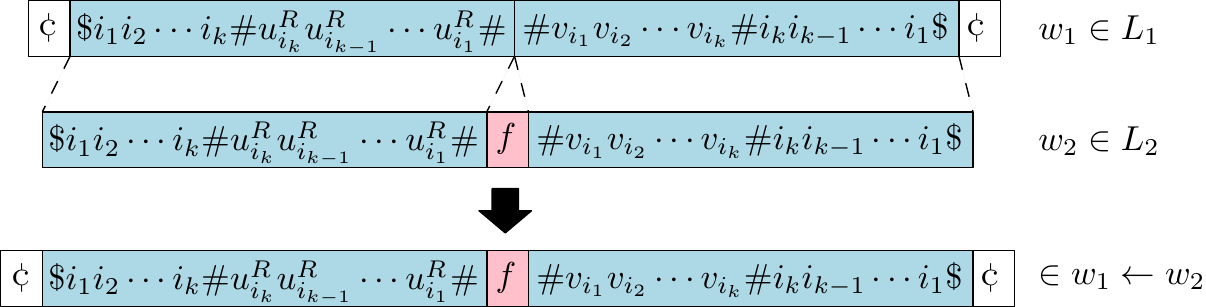}
\centering
\caption{Only possible outfix-guided insertion
of $w_2$ into $w_1$.}
\label{fatta1}
\end{figure}

\item[{\rm $\bullet$ ``$I_{\rm PCP}$ has no solution implies
$L$ is  closed'':}]  
Recall that by a trivial outfix-guided
derivation step we mean a derivation step
$w \stackrel{[w]}{\Rightarrow} w$ where $w$ is obtained from
itself by selecting the outfix to consist of a prefix and suffix
of $w$ whose concatenation is equal to $w$.

Using the assumption that the instance $I_{\rm PCP}$ does not
have a solution we show that strings of $L$ can be inserted
into strings of $L$ using only trivial derivation steps which
naturally then implies $(L \leftarrow L) \subseteq L$.

Strings of $L_1$ begin and end with the symbol $\cent$ and this
symbol occurs exactly two times in strings of $L_1$. Thus, strings
of $L_1$ can be outfix-guided inserted into strings of $L_1$
only using a trivial derivation step. For the same reason
(by replacing $\cent$ with \$)
strings of $L_2$ can be inserted into strings of $L_2$ only
using a trivial derivation step.

Strings of $L_1$ cannot be outfix-guided inserted into strings of $L_2$
because the former begin and end with the symbol $\cent$ and the latter
do not contain any occurrences of $\cent$.
The remaining possibility we need to consider is under what
conditions strings of $L_2$ can be outfix-guided inserted into 
strings of $L_1$.

Consider 
$$w_1 = 
\cent \$ i_1 i_2 \cdots i_r \# u_{i_r}^R u_{i_{r-1}}^R \cdots
u_{i_1}^R \# \# v_{j_1} v_{j_2} \cdots v_{j_s} \# j_s j_{s-1} \cdots
j_1 \$ \cent \in L_1, 
$$
$$
\mbox{ and, }
w_2 = \$ i_1 i_2 \cdots i_r \# w \# f \# w^R \# i_r
i_{r-1} \cdots i_1 \$ \in L_2,
$$
and suppose we can write $w_1 = x_1 u v x_2$, $w_2 = uzv$, 
$u, v \neq \varepsilon$. Since $w_1$ does not contain occurrences
of the symbol $f$, in the decomposition of $w_2$ the symbol $f$ must
be in the substring $z$. The string $w_2$ begins and ends with
$\$ $ which are then the first and last symbol of $u$ and $v$, 
respectively (and consequently $x_1 = x_2 = \cent$).
Since the concatenation
of $u$ and $v$ must contain all four occurrences of \# in $w_1$,
it follows that 
$$
u = \$ i_1 i_2 \cdots i_r \# u_{i_r}^R u_{i_{r-1}}^R \cdots
u_{i_1}^R \# \; \mbox{ and } \;
v = \# v_{j_1} v_{j_2} \cdots v_{j_s} \# j_s j_{s-1} \cdots
j_1 \$.
$$
Now in the decomposition of $w_2$ the only possibility is that $z = f$,
and the strings $w_1$ and $w_2$ must be as illustrated
in Fig.~\ref{fatta2}.
From the definition
of the language $L_2$ it follows that $r = s$, $i_x = j_x$, 
$x = 1, \ldots, r$. 
and $( u_{i_r}^R u_{i_{r-1}}^R u_{i_1}^R )^R = 
v_{i_1} v_{i_2} \cdots v_{i_r}$. Together these conditions mean
that $(i_1, \ldots, i_r)$ is a solution for the
instance $I_{\rm PCP}$, which contradicts our assumption that
the instance did not have a solution.
\begin{figure}[htb]
\includegraphics[width=12.5cm]{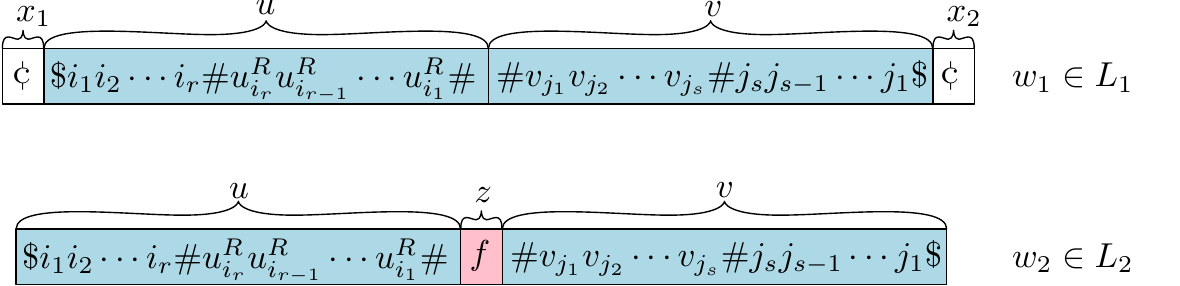}
\centering
\caption{Decompositions of $w_1\in L_1$ and $w_2\in L_2$.}
\label{fatta2}
\end{figure}
\end{description}
\prend
\end{proof}

Note that in the proof of Theorem~\ref{tatta52} the language
$L_1$ is not deterministic context-free. 
It remains open whether og-closure can be
decided for deterministic context-free languages.

\section{Conclusion}

Analogously with the recent overlap assembly 
operation~\cite{Csuhaj2007,EnagantiIKK15},
we have introduced an overlapping insertion operation on strings
and have studied
closure and decision properties of the outfix-guided
insertion operation.  While closure properties of non-iterated
outfix-guided
insertion are straightforward to establish, the questions become
more involved for the outfix-guided insertion closure.
As the main result we have shown that the
outfix-guided insertion closure of a finite language need not
be regular.

Much work remains to be done on outfix-guided insertion. One
of the main
open questions is to determine upper bounds for the complexity
of the outfix-guided insertion closures of regular languages. Does there
exist regular languages $L$ such that the outfix-guided insertion
closure of $L$ is non-context-free?

\section*{Acknowledgments}

We thank the referees for many useful suggestions that have
improved the presentation of the paper.
Cho and Han were supported by the Basic Science Research Program
through NRF funded by MEST (2015R1D1A1A01060097), the Yonsei
University Future-leading Research Initiative of 2016 and the IITP
grant funded by the Korea government (MSIP) (R0124-16-0002).
 Ng and Salomaa were
supported by Natural Sciences and Engineering Research
Council of Canada Grant OGP0147224.

\bibliographystyle{plain}
\bibliography{OutfixGuided}

\end{document}